\documentclass[twoside]{article}
\usepackage[a4paper]{geometry}
\usepackage[latin1]{inputenc} 
\usepackage[T1]{fontenc} 
\usepackage{RR}
\usepackage{hyperref}

\usepackage{graphicx}
\usepackage{amssymb}
\usepackage{amsthm}
\usepackage{amsmath}
\usepackage[english]{babel}
\usepackage[latin1]{inputenc}
\usepackage{times}
\usepackage{import}
\usepackage[T1]{fontenc}
\usepackage{verbatim}
\usepackage{bm}
\usepackage{algorithm}
\usepackage{algorithmic}

\newtheorem{definition}{Definition}
\newtheorem{theorem}{Theorem}

\newtheorem{lemma}{Lemma}
\newtheorem{corollary}{Corollary}

\RRNo{7985}
\RRdate{May 2012}
\RRauthor{
Chuanyou Li
\thanks{
PhD student invited in the research team CIDre from December 2011 to November 2012.
}
\thanks[cn1]{
School of Computer Science and Engineering, Southeast University, Nanjing, China.
}
\thanks[cn2]{
Key Lab of Computer Network \& Information Integration, Ministry of Education.
}
\and
Michel Hurfin
\thanks{
INRIA Rennes Bretagne Atlantique, EPI CIDre, Rennes, France.
}
\and
Yun Wang
\thanksref{cn1}
\thanksref{cn2}
 }
\authorhead{Li \& Hurfin \& Wang}
\RRtitle{
R\'{e}soudre le probl\`{e}me du consensus approximatif\\
en pr\'{e}sence de Byzantins\\
dans un r\'{e}seau mobile partiellement connect\'{e}\\
}
\RRetitle{
Reaching Approximate Byzantine Consensus\\
 in Partially-Connected Mobile Networks
}
\titlehead{Reaching Approximate Byzantine Consensus in Partially-Connected Mobile Networks}
\RRresume{
Nous consid\'{e}rons le probl\`{e}me du consensus approximatif dans des r\'{e}seaux mobiles contenant des n{\oe}uds byzantins. Nous supposons que chaque n{\oe}ud correct ne peut communiquer qu'avec ses voisins et n'a pas connaissance de la topologie globale. Comme tous les n{\oe}uds ont la possibilit\'{e} de se d\'{e}placer, la topologie est dynamique. Le nombre de n{\oe}uds byzantins est born\'{e} par $f$ et est connu de tous les no{\oe}uds corrects. Nous pr\'{e}sentons tout d'abord un protocole de consensus approximatif byzantine qui est fond\'{e} sur la m\'{e}thode d'it\'{e}ration lin\'{e}aire. Comme les n{\oe}uds sont autoris\'{e}s \`{a} collecter des informations lors de plusieurs tours cons\'{e}cutifs, le fait de se d\'{e}placer leur donne l'occasion de recueillir plus de valeurs. Nous proposons une nouvelle condition n\'{e}cessaire et suffisante pour garantir la convergence finale du protocole de consensus. La contrainte exprim\'{e}e par notre condition n'est pas "universelle": lors de chaque phase, elle ne concerne qu'un seul n{\oe}ud correct. Plus pr\'{e}cis\'{e}ment, au moins un no{\oe}ud correct parmi ceux qui proposent  la valeur minimale ou la valeur maximale pr\'{e}sente dans le r\'{e}seau, doit recevoir suffisamment de messages (contrainte sur la quantit\'{e}) contenants des valeurs sup\'{e}rieures ou inf\'{e}rieures (contrainte sur la la qualit\'{e}). Bien entendu, les d\'{e}placements des no{\oe}uds doivent permettre \`{a} cette condition d'\^{e}tre remplie. Notre conclusion montre que la condition propos\'{e}e peut \^{e}tre satisfaite si le nombre total de n{\oe}uds est plus grand que $3f +1$.
}

\RRabstract{
We consider the problem of approximate consensus in mobile networks containing Byzantine nodes. We assume that each correct node can communicate only with its neighbors and has no knowledge of the global topology. As all nodes have moving ability, the topology is dynamic. The number of Byzantine nodes is bounded by $f$ and known by all correct nodes. We first introduce an approximate Byzantine consensus protocol which is based on the linear iteration method. As nodes are allowed to collect information during several consecutive rounds, moving gives them the opportunity to gather more values. We propose a novel sufficient and necessary condition to guarantee the final convergence of the consensus protocol. The requirement expressed by our condition is not "universal": in each phase it affects only a single correct node. More precisely, at least one correct node among those that propose either the minimum or the maximum value which is present in the network, has to receive enough messages (quantity constraint) with either higher or lower values (quality constraint). Of course, nodes' motion should not prevent this requirement to be fulfilled. Our conclusion shows that the proposed condition can be satisfied if the total number of nodes is greater than $3f+1$.
}
\RRmotcle{
probl\`{e}me d'accord,
consensus approximatif,
faute Byzantine,
syst\`{e}me r\'{e}parti,
condition n\'{e}cessaire et suffisante,
topologie dynamique,
mobilit\'{e},
r\'{e}seau ad-hoc
}
\RRkeyword{
agreement problem,
approximate consensus,
Byzantine fault,
distributed system,
necessary and sufficient condition,
dynamic topology,
mobility,
ad-hoc network
}

\RRprojet{Cidre}
\RCRennes 

\begin{document}
\makeRR   


\section{Introduction}
We consider a distributed system where nodes are mobile and form an ad hoc network characterized by a dynamic topology. When a node changes its physical location by moving around, it also changes the set of its neighbors with whom it can communicate directly (roughly speaking, nodes that are physically nearby). The system is unreliable. Nodes may suffer from Byzantine faults and messages may be lost. A Byzantine node, also called a malicious node, may stop its activity or execute arbitrary code. In particular, it may send messages with fake values. Nodes that are not malicious are said to be correct.

Consensus is recognized as a basic paradigm for fault-tolerance in distributed systems. According to the application's needs, several variants of the  consensus problem have been proposed. Among these agreement abstractions, one is called the {\it Approximate consensus} problem and has been presented for the first time in~\cite{ex1}. Each node begins to participate by providing a real value called its initial value. Eventually all correct nodes must obtain final values  that are different from each other within a maximum value denoted $\epsilon$ (convergence property) and must be in the range of initial values proposed by the correct nodes (validity property). Approximate consensus can be used in applications (clock synchronization, distributed data fusion,$\ldots$) that do not require to achieve exact agreement on a single outcome value.

Several protocols have been proposed to solve this problem in the presence of Byzantine nodes.
Some protocols~\cite{ex1,ex6} assume that the network is fully connected: during the whole execution, a correct node should be able to communicate by message passing with any other correct node. Obviously, this property is not satisfied in our context. Other protocols~\cite{ex2,ex3,ex4} consider partially connected networks but require an additional constraint: any correct node must know the whole topology.  Again, such a global information is impossible to obtain in our context. Based on the linear iterative consensus strategy~\cite{ex16}, recent protocols~\cite{ex7,ex13,ex15} also assume that the network is partially connected but do not require any global information. At each iteration, a correct node broadcasts its value, gathers values from its neighborhood and updates its own value. Its new value is an average
of its own previous value and those of some of its neighbors. Like in~\cite{ex1}, before computing its new value, a correct node must ignore some of the values it has collected. These removed values may have been proposed by Byzantine nodes and may invalidate the validity property. In order to achieve convergence, the proposed solutions rely on additional conditions that have to be satisfied by the topology. In~\cite{ex7,ex15}, the proposed conditions are proved to be sufficient and necessary in the case of an arbitrary directed graph.

The solution presented in this paper addresses the approximate Byzantine consensus problem  in Partially-Connected Mobile Networks. It follows the general strategy proposed in~\cite{ex7,ex13,ex15}. However, it differs from these previous works for two main reasons.  First we modify the iterative protocol to cope more efficiently with mobility. Each node still follows an iteration scheme and repeatedly executes rounds. Yet a round is now decomposed into two parts: a moving step followed by a computing step. Furthermore, during the computing step, a node still broadcasts its value, gathers values and updates its values but now the values used to compute  its new value have not necessarily been received during the current round. In other words, a correct node can now take into account values contained in messages sent during consecutive rounds. An integer parameter (denoted $R_{c}$ hereafter) is used to fix the maximal number of rounds
during which values can be gathered and stored while waiting to be used.  Thanks to this flexibility, a node can use its ability to travel to collect enough values. The second difference is the most important one. While the solutions proposed in~\cite{ex7,ex13,ex15} define conditions that refer only to the topology, we present a condition that considers also the values proposed by correct nodes. To understand the interest of our approach, let us consider the following example.
One correct node $p_{i}$ proposes an initial value $v_{a}$ while all the other correct nodes propose an initial value $v_{b}$. In this particular scenario, if the node $p_{i}$ can receive values from a sufficient number of correct neighbors, approximate consensus can be reached even if all the other nodes are isolated and receive no message. This example suggests that the location of values is just as important as the network topology. In~\cite{ex7,ex15}, constraints on the topology ensure that each node has enough neighbors. These constraints are "universal" because they affect all nodes in the network.  In a mobile environment, it is difficult to ensure that no node is never isolated from (or insufficiently connected to) the rest of the network. Furthermore, the above example shows that a strong universal constraint is not always necessary. In this paper, a novel sufficient and necessary condition is proposed. In this condition, topology and values proposed by correct nodes are both taken into account. The condition affects only a subset of nodes that can change from one round to another. More precisely, the condition focuses only on the correct nodes that propose either the maximum or the minimum value and imposes no obligation on the other nodes. To achieve consensus,  from time to time, at least one of these particular nodes must receive enough messages (quantity requirement)  with values different from its current value (quality requirement). Obviously, the constraint  is weaker and not universal as it has to be satisfied by a single node.

The rest of this paper is organized as follows. Section~\ref{sec:model} introduces the model and provides a formal definition of the approximate consensus problem. In Section~\ref{sec:protocol},  we present our protocol based on linear iteration and we prove that correct nodes will never violate the validity property by adopting illegal values. Section~\ref{sec:relatedworks} sketches out some related works. To ensure convergence, Section~\ref{sec:condition} proposes a sufficient condition. Then this condition is slightly modified to obtain a sufficient and necessary version.  Section~\ref{sec:conclusion} brings our concluding remarks.

\section{Model and Problem Definition }
\label{sec:model}

\subsection{Model}
We consider a mobile distributed system composed of $n$ nodes $V=\{p_{1}, p_{2},...,p_{n}\}$. During the entire period of computation, each node $p_{i}$ can move towards any direction and at any speed within a limited geographical area. Nodes communicate with each other only by exchanging messages. A node can only communicate with its close neighbors. Therefore the topology ({\it i.e.}, the communication graph) is dynamic. When receiving a message, the receiver knows the correct identity of the sender. The communication is synchronous. Messages can be lost but there are no duplicate messages and each channel is FIFO.

Nodes are divided into two subsets denoted $C_{n}$ and $F_{n}$. The set $C_{n}$ contains the correct nodes which always follow the protocol's specification. The nodes of the set $F_{n}$ are Byzantine nodes. They behave arbitrarily and can collude together. In particular, each of them can stop its computation or send messages with different fake values to different neighbors. No assumption restricts their possible behaviors. However, the total number of Byzantine nodes is limited by $f$.

The protocol described in Section~\ref{sec:protocol} is based on an iterative process. A sequence of rounds is carried out by each node. A round is identified by a round number $r$ that belongs to the set $R=\{1,2...\}$. For simplicity, the schedulers of all correct nodes are assumed to be fully synchronous. Each round $r$ is divided into two parts denoted $r_{m}$ (mobility part) and $r_{p}$ (protocol part). During $r_{m}$, a node can either move to a new location or stay in the same place.  During $r_{p}$, a node $p_{i}$ broadcasts its value $v$, gathers values, and updates its state: a consensus protocol (such as the one proposed in Section 3) describes the computation performed by $p_{i}$ during a round.

\begin{figure}[ht]
    \centering
    \includegraphics[width=2.5in]{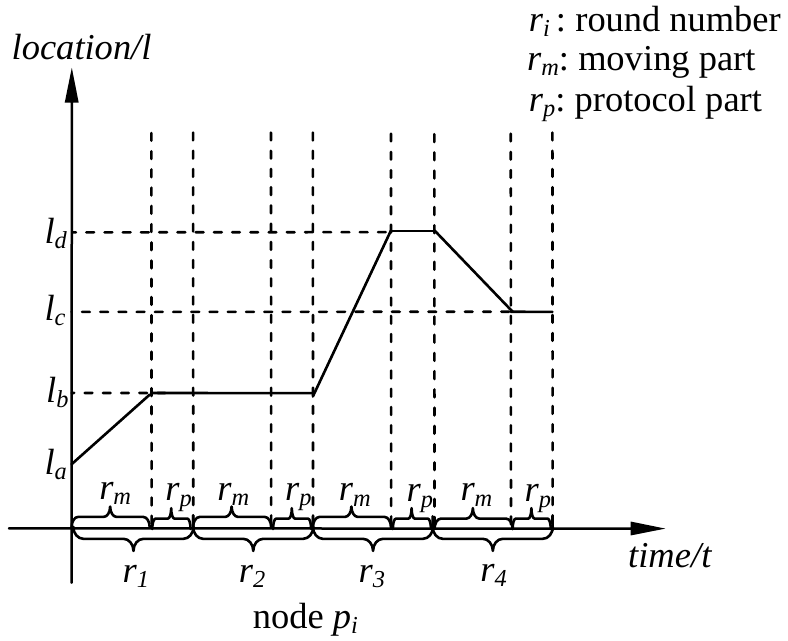}
    \caption{Path followed by node $p_{i}$ during 4 rounds}
    \label{Fig1}
\end{figure}

The behavior of a correct node $p_{i}$ during 4 consecutive rounds is described in Figure~\ref{Fig1}.
Just before executing round $r_{1}$, $p_{i}$ is located in $l_{a}$. During the first part of round $r_{1}$, $p_{i}$ moves to another location $l_{b}$ and executes the protocol. The node $p_{i}$ remains in location $l_{b}$ during round $r_{2}$ and executes again the protocol. It moves to $l_{d}$ during round $r_{3}$ and executes the protocol. In round $r_{4}$, $p_{i}$ moves to location $l_{c}$ and it executes the protocol for the fourth time.

During a round $r$, a {\it simple directed graph} $G_{r}(V, E_{r})$ is used to model the dynamic topology. If during round $r$, node $p_{i}$ can receive a message from a node $p_{j}$ located in its neighborhood then there is a directed link from $p_{j}$ to $p_{i}$: $(p_{j},p_{i})\in E_{r}$. In the proposed protocol, a correct node $p_{i}$ can receive a value during round $r$, keep it during several consecutive rounds, and use it in a future round $r+k$. Therefore, the concept of {\it joint graph}~\cite{ex5} is also used within this paper. A joint graph is defined as the union of the graphs corresponding to several well-identified consecutive rounds.  Figure~\ref{Fig2} illustrates this concept in the particular case of two consecutive rounds $r_{1}$ and $r_{2}$. The graphs $G_{r_{1}}(V, E_{r_{1}})$ and $G_{r_{2}}(V, E_{r_{2}})$ are depicted on the left side. The corresponding  joint graph appears on the right side.

\begin{figure}[ht]
    \centering
    \includegraphics[width=3.0in]{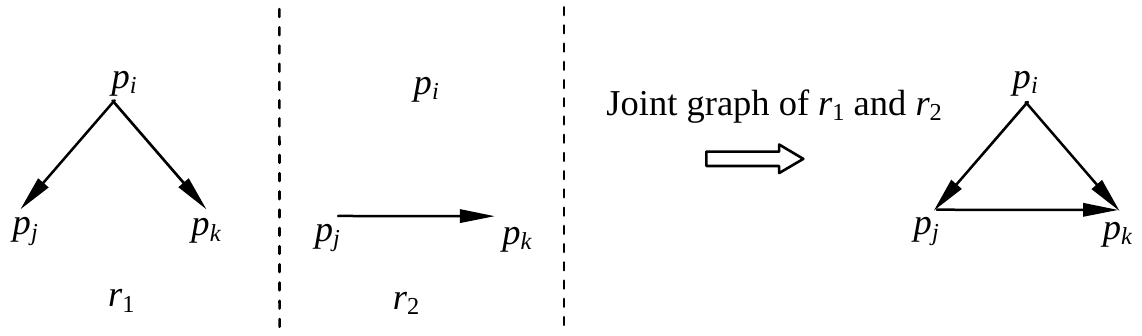}
    \caption{Two graphs and the associated joint graph}
    \label{Fig2}
\end{figure}

\subsection{Definition of the Agreement Problem}

Within this paper, the value of a correct node $p_{i}$ at the beginning of round $r$ is denoted $v_{i}(r)$. Consequently the initial value of $p_{i}$ is denoted $v_{i}(1)$. The minimum (respectively maximum) value proposed by correct nodes during round $r$ is denoted $v_{min}(r)$ (respectively $v_{max}(r)$).

\begin{definition}
{\rm The approximate Byzantine consensus problem is formally defined by two properties:}\\
\noindent
Validity property:\\
{\rm During any round $r$, the value of a correct node is in
the range of initial values of correct nodes:}\\
$\forall p_{i} \in C_{n}$, $\forall r \geq 1$, $v_{i}(r)\in [v_{min}(1), v_{max}(1)]$\\
\noindent
Convergence property: \\
{\rm Eventually all correct nodes have values which are different from each other within a maximum predefined value denoted $\epsilon$ and such that $\epsilon > 0$.}\\
$\forall p_{i},p_{j}  \in C_{n}$, $\exists N > 0$,  $\forall r > N$,  $\mid v_{i}(r) - v_{j}(r) \mid < \epsilon$
\end{definition}

\section{The Protocol and its Safety Proof}
\label{sec:protocol}

\begin{algorithm}[!t]
    \small
    \caption{Linear Approximate Byzantine Consensus}
    \begin{algorithmic}[1]
        \STATE $r$ $\leftarrow$ 1;
        \STATE $v_{i}(1) \leftarrow$ the initial value proposed by $p_{i}$;
        \STATE $Neb_{i}(1)\leftarrow null$;
        \\ $~$
        \STATE \textbf{for any node $p_{i}$ in round $r$};
        \STATE \textbf{do};
        \STATE $p_{i}$ sends $v_{i}$ to its neighbors;
        \STATE $p_{i}$ waits for receiving messages;
        \STATE $Neb_{i}(r)\leftarrow$ \{new values from neighbors\}$\cup Neb_{i}(r)$;
        \STATE $Neb_{i}(r)\leftarrow$sort($Neb_{i}(r)$);
        \STATE $x \leftarrow$ the number of values bigger or equal than $v_{i}(r)$;
        \STATE $y \leftarrow$ the number of values less or equal than $v_{i}(r)$;
        \IF {($x\geq f+1$ or $y\geq f+1$)}
            \STATE $Neb_{i}(r)\leftarrow$reducing($Neb_{i}(r)$, $f$, $x$, $y$);
            \STATE $v_{i}(r+1)\leftarrow$ average($Neb_{i}(r), v_{i}(r)$);
            \STATE $Neb_{i}(r+1) \leftarrow$ $null$;
        \ELSE
            \STATE $v_{i}(r+1)\leftarrow v_{i}(r)$;
            \IF{($r$ mod $R_{c}$ equals to 0)}
                \STATE $Neb_{i}(r+1) \leftarrow$ $null$;
            \ELSE
                \STATE $Neb_{i}(r+1)\leftarrow Neb_{i}(r)$;
            \ENDIF
        \ENDIF
        \STATE $r\leftarrow r+1$;
        \STATE \textbf{enddo};
        \\ $~$
        \STATE Procedure reducing($Neb_{i}(r)$), $f$, $x$, $y$);
        \STATE \textbf{do};
        \STATE $B\leftarrow$ the set of $f$ largest values in $Neb_{i}(r)$;
        \STATE $S\leftarrow$ the set of $f$ smallest values in $Neb_{i}(r)$;
        \IF {$x > y$}
            \STATE Suppress all the values of $B$;
            \STATE Suppress the values $v_{j}\in S$ such that $v_{j}<v_{i}(r)$;
        \ELSE
            \STATE Suppress all the values of $S$;
            \STATE Suppress the values $v_{j}\in B$ such that $v_{j}>v_{i}(r)$;
        \ENDIF
        \STATE \textbf{enddo};
        \\ $~$
        \STATE Procedure average($Neb_{i}(r), v_{i}(r)$)
        \STATE \textbf{do};
        \STATE $n_{i}\leftarrow \mid Neb_{i}(r)\mid$;
        \STATE $v_{new}\leftarrow \frac{v_{i}(r)+\sum_{j}v_{j}}{n_{i}+1}$, $v_{j}\in Neb_{i}(r)$;
        \STATE return $v_{new}$;
        \STATE \textbf{enddo};
    \end{algorithmic}
\end{algorithm}

\subsection{An Iterative Protocol}

The pseudo-code (See Algorithm 1) is executed by all the correct nodes  during the second part of each round $r \geq 1$.

The execution of the three first lines initializes the three main variables managed by a node $p_{i}$: its current round number $r$, its current value $v_{i}(r)$ and a multi-set $Neb_{i}$ which is used to store values received from neighbors. From time to time, $Neb_{i}$ is reset to \emph{null}, in accordance with a strategy explained later. The rest of the code is divided into two main stages called {\it gathering} (line 6-11) and {\it updating} (line 12-24).

Node $p_{i}$ and its neighbors exchange their values (line 6-7). The received values are logged into the multi-set $Neb_{i}$ (line 8). A received value can be kept in $Neb_{i}$ during at most $R_{c}$ rounds (See the test at line 18). During a round $r$, $p_{i}$ receives at most one value from each (correct or Byzantine) node. But if $Neb_{i}$ has not been reset for several rounds, $p_{i}$ can receive a value $v_{j}(r)$ during round $r$ while a value $v_{j}(r-k)$ previously provided by the same node $p_j$ is already in $Neb_{i}$. In that case, $p_{i}$ keeps only the most recent value.  When $p_{i}$ stops collecting values, all values of $Neb_{i}$ are sorted into ascending order (line 9).

To guarantee the validity property, a correct node $p_{i}$ must gather enough values to be allowed to compute a new value (line 14). Otherwise, $p_{i}$ has to start the next round with the same value (line 17). During round $r$, the test evaluated by $p_{i}$ at line 12 defines two favorable cases: either $p_{i}$ has received values that are greater than or equal to $v_{i}(r)$ from at least $f+1$ different nodes, or $p_{i}$ has received values that are smaller than or equal to $v_{i}(r)$ from at least $f+1$ different nodes. If $p_i$ has received less that $f+1$ values from different nodes, the test cannot be satisfied. If $p_{i}$ has received values from at least $2f+1$ different nodes, the test is necessarily satisfied.  If $p_{i}$ has gathered more than $f$ values but less than $2f+1$ values, the test can be satisfied or not.

\begin{figure}[hb]
    \centering
    \includegraphics[width=2.8in]{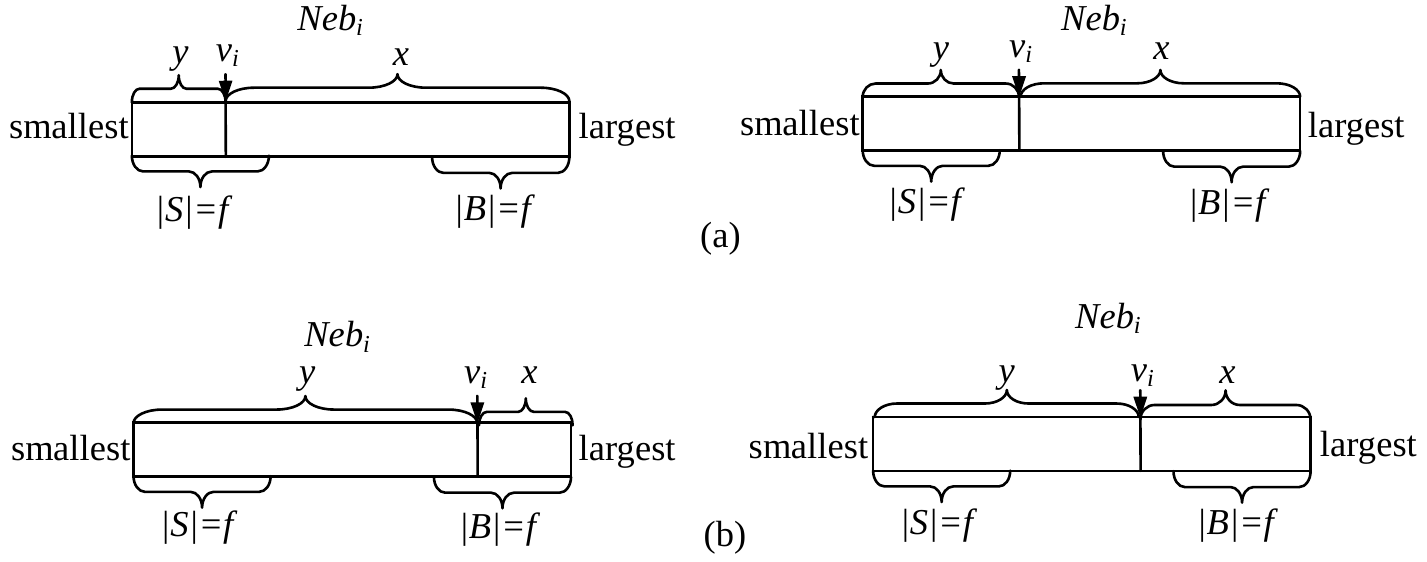}
    \caption{The reducing procedure}
    \label{Fig3}
\end{figure}

When the test of line 12 is satisfied, $p_{i}$ executes sequentially the {\it reducing} procedure (line 13) and the {\it average} procedure (line 14). Reducing operation has been introduced in~\cite{ex1}. To ensure the validity property, a few values have to be removed from the multi-set $Neb_{i}(r)$. The strategy used in this paper leads to suppress between $f$ and $2f$ values while the strategy used in~\cite{ex1} leads to always ignore exactly $2f$ values.  To choose the removed values, $p_{i}$ compare the received ones with its own current value $v_{i}$. Within the set of values $Neb_{i}(r)$, $B$ is defined as the subset that contains the $f$ largest values (line 28) while $S$ is defined as the subset that contains all the $f$ smallest values (line 29). Due to the fact that the test of line 12 is satisfied, either at least $f$ values in $Neb_{i}$ are greater than or equal to $v_{i}$ (case a) or at least $f$ values in $Neb_{i}$ are smaller than or equal to $v_{i}$ (case b). The two cases (a and b) are depicted in Figure~\ref{Fig3} where the sorted set $Neb_{i}$ is represented by a rectangle. Note that the two cases are not mutually exclusive. Thus, the two representations of $Neb_{i}(r)$ that appear on the right side of Figure~\ref{Fig3} are equivalent and may lead to suppress $2f$ values ({\it i.e.} all the values of $B$ and $S$) if $v_{i}$ belongs neither to $B$ nor $S$.  Less values will be removed if we consider the two representations on the left side. In case a, only the values of $B$ and the values $v_{j}$ of $S$ such that $v_{j} < v_{i}(r)$ are suppressed from $Neb_{i}$. In case b, only the values of $S$ and the values $v_{j}$ of $B$ such that $v_{j} > v_{i}(r)$ are removed.

After reducing, $p_{i}$ executes the average procedure.  $p_{i}$ considers only the remaining values of $Neb_{i}(r)$. It calculates average with the values in $Neb_{i}(r)$ and $v_{i}(r)$. The weight is simply set to $\mid Neb_{i}(r)\mid+1$ and $v_{i}(r+1)$ is assigned to the computing result. Then $Neb_{i}(r+1)$ is reset to $null$ (line 15).

\subsection{Resetting the Log of Values $Neb_{i}$}

The variable $Neb_{i}$ is initialized to {\it null} (line 3) and can be reset to {\it null} in two different cases (line 15 and 19). As indicated in the previous paragraph, all the gathered values are suppressed when the node $p_{i}$ computes a new value during the round r (line 15). Therefore, in the future, this node will only use values issued during a round higher than $r$.

In a mobile environment, a node is sometimes isolated or at least weakly connected to the rest of the network. The number of values collected during a given round is sometimes low. If a reset is made systematically at the end of each round, the test of line 12 (used to control if enough values have been gathered) is rarely satisfied. By reducing the frequency of reset operations, a node can collect more values over several consecutive rounds. Thus it may take advantage of mobility to increase the number of discovered neighbors. Consequently, the probability that it can frequently calculate a new value increases. Yet the reset operation is very important and is a key element in the proof of the convergence property. A periodic reset operation cleans the system of old values. If a mobile node $p_{i}$ moves far away and keep some very old values in $Neb_{i}$ for a long period of time (which is not bounded by a number of rounds), a negative impact on the convergence can be observed.

In the proposed protocol, a general reset is performed by all the correct nodes every $R_{c}$ rounds (line 18 and 19). By construction, line 19 is executed during a round $r$ such that $r = k R_{c}$ with $k \geq 1$. The execution of line 3 during the initialization phase can also be considered as a general reset performed during a fictive round numbered 0. We define the set $S^{c}$ as the set of all the rounds $r_c$ such that the instruction "$Neb_{i}(r_c) \leftarrow null$" as been executed either at line 3 or at line 19 of round $r_c - 1$. These rounds are called {\it common new starting} rounds. By definition, $S^{c} = \{ r_c \mid r_c = k R_c + 1 {\rm ~with~} k \geq 0\}$. The common new starting rounds allow to divide the computation into phases. Each phase is identified by the value of the integer $k$  and is composed of $R_c$ rounds.
We define also the concept of  {\it local new starting round} as follows. From the point of view of a correct node $p_{i}$, $r$ is a local new starting round, if either $r=1$ or  $p_{i}$ resets $Neb_{i}$ to \emph{null} during round $r-1$. The set of all local new starting rounds of $p_{i}$ is denoted $S_{i}$. Obviously, for any correct node $p_{i}$, $S^{c} \subseteq S_{i}$. By definition, during a phase, a correct node $p_{i}$ executes a reset operation at least once (during the last round of the phase) and at most $R_{c}$ times (each time a new value is computed). The local new starting rounds of a correct node $p_{i}$ are used to identify some particular joint graphs (See Section~\ref{sec:model}). Let $r$ be a round executed by $p_i$. By definition, there exist a unique local new starting round $r_{s_1} \in S_{i}$ such that $r_{s_1} \leq r$  and for any $r_{s_2} \in S_{i}$ either
$r_{s_2} \leq r_{s_1}$ or $r < r_{s_2}$. Round $r_{s_1}$ is $p_{i}$'s latest new starting round.
The joint graph corresponding to the union of the communication graphs observed during the non empty sequence of consecutive rounds beginning with $r_{s_1}$ and ending with $r$ is used to identify the nodes which have communicated their values to $p_{i}$ during this period.
In this paper, the notation $JN_{i}^{r}$ is used to represent the \emph{joint neighbor set} of $p_{i}$ at round $r$.

\subsection{Validity Property and Legal Values}

\begin{theorem}
The proposed protocol satisfies the \emph{validity} property.
\end{theorem}

\begin{proof}
Obviously, the  property is satisfied during the first round: $\forall p_{i} \in C_{n}, v_{i}(1) \in [v_{min}(1), v_{max}(1)]$. Let us consider that the property is satisfied during any round smaller than or equal to $r$. To violate the property during round $r+1$, at least one correct process $p_{i}$ must modify its value during the execution of the average procedure and must adopt a new value with is either smaller than $v_{min}(1)$ or greater than $v_{max}(1)$. Due to the properties of the average function, at least one value that is either smaller than $v_{min}(1)$ or greater than $v_{max}(1)$ must appear in the multi-set $Neb_{i}(r)$. A value $v$ contained in this set is either proposed by a correct node or by a Byzantine node. In the first case, due to the induction assumption, $v$ belongs to the range $[v_{min}(1), v_{max}(1)]$. In the second case, $v$ can remains in $Neb_{i}(r)$ after the execution the reducing procedure only if  at least $f+1$ fake values have been gathered. As the number of Byzantine nodes is bounded by $f$, the \emph{validity} property is always satisfied.
\end{proof}

Some works~\cite{ex7} adopt a property which is stronger than the above validity property. During the whole computation, the maximum value proposed by a correct node has to be  monotonically non-increasing and similarly the minimum value has to be monotonically non-decreasing. More precisely, for any round $r \geq 1$, the conditions $v_{min}(r) \leq v_{min}(r+1)$ and $v_{max}(r) \geq v_{max}(r+1)$ must hold.
The proposed protocol can satisfy this stronger property if and only if $R_{c} = 1$. When $R_{c} > 1$, as a correct node $p_{i}$ may keep old values in its log $Neb_{i}(r)$, the above conditions are not always true. The new value computed by $p_{i}$ during round $r$, namely $v_{i}(r+1)$, may be less than $v_{min}(r)$ or bigger than $v_{max}(r)$. To take this possibility into account, we define first the concept of \emph{legal value} and then we propose a {\it safety} property which is stronger than our original {\it validity} property.

\begin{definition}
Let $r$ by a round number such that $r = kR_{c}+m$ with $k\geq0$ and $1 \leq m \leq R_c$.
The value $v_{i}(r)$ of a correct node $p_{i}$ is legal if the two conditions $v_{i}(r) \geq v_{min}(d)$ and $v_{i}(r) \leq v_{max}(d)$ are satisfied when the round number d is defined as follows:
\begin{enumerate}
\item $(k = 0) \wedge (m = 1)$: $d = 1$
\item $(m \neq 1)$: $d = kR_{c}+1$
\item $(k \neq 0) \wedge (m = 1)$: $d = (k-1)R_{c}+1$
\end{enumerate}
\end{definition}

\begin{lemma}
$\forall p_{i} \in C_{n}$, $\forall r \geq 1$, $v_{i}(r)$ is legal.
\end{lemma}

\begin{proof}
Depending on the round number $r = kR_{c}+m$, three cases that are mutually exclusive have to be considered. When $m=1$ and $k=0$, the value $v_{i}(1)$ of a correct process $p_{i}$ is in the range $[v_{min}(1), v_{max}(1)]$. In the two remaining cases, we prove that $v_{i}(r) \leq v_{max}(d)$. A similar demonstration can be done to conclude that $v_{i}(r) \geq v_{min}(d)$.

If  $m \neq 1$, then $r$ is not a common new starting round. The nearest previous common new starting round is $d = kR_{c}+1$. As $m>1$, we have $r > d$. By definition, at the beginning of round $d$, every correct node $p_{j}$ has no value in its set $Neb_{j}(d)$. Furthermore, at that time, for any correct node $p_{j}$, the property $v_{j}(d) \leq v_{max}(d)$ holds. Now the proof is by contradiction. Let us consider that $r$ is the very first round greater than $d$ during which at least one correct node $p_{i}$ violates the property. Thus, we have $v_{i}(r) > v_{max}(d)$. The computation of the value $v_{i}(r)$ has been done by $p_{i}$ during the previous round $r-1$. All the values used during the execution of the average procedure by $p_{i}$ have been received by $p_{i}$ during round $r-1$ and may be during rounds $r-2$, $\ldots$,$d+1$ and $d$. In all the possible cases, any value $v$ received from a correct node is such that $v \leq v_{max}(d)$. To have still a value greater than $v_{max}(d)$ and thus greater than $v_{i}(r-1)$ in its log $Neb_{j}(r-1)$ after the execution of the reducing procedure, $p_{i}$ must gather $f+1$ fake values. This contradicts both the fact that the network contains at most $f$ Byzantine nodes and the fact that a node (correct or not) cannot insert two different values in the multi-set $Neb_{i}$ of a correct node $p_{i}$.

If $m=1$ and $k>0$, then $r$ is a common new starting round. The value $v_{i}(r)$ has been computed by $p_{i}$ during the round $r-1 = kR_{c}$ and $Neb_{i}(r-1)$ may contain values proposed during the $R_{c}$ previous rounds. As the round $d = (k-1)R_{c}+1$ is also a common new starting round, a similar reasoning leads to conclude that $v_{i}(r) \leq v_{max}(d)$.
\end{proof}

Note that after the execution of the reducing procedure , all the remaining values are legal.  
The following corollary focuses on the common new starting rounds that identify the beginning of phases. This corollary can be considered as our new {\it safety} property. 

\begin{corollary}
Safety property: \\
$\forall r \in S^{c}$, $\forall p_{i} \in C_{n}$, $\forall x \geq 0$,  $v_{i}(r+x) \geq v_{min}(r)$ and $v_{i}(r+x) \leq v_{max}(r)$
\end{corollary}

\begin{proof}
As $r \in S^{c}$, there exists an integer $k \geq 0$  such that $r= kRc+1$. When $x = R_c$, we have $r+x = (k+1)Rc +1$. Due to lemma 1, we conclude directly that the two conditions holds. When $x$ is a multiple of $R_c$, the proof is also obvious. Finally, when $x = k'R_c+m'$ with $1 \leq m' \leq R_c -1$, we have $r+x = (k+k')Rc +m'+1$. Again the proof relies on Lemma 1.
\end{proof}

\section{Related works}
\label{sec:relatedworks}

Dolev et al. are the first to address the approximate consensus problem in the presence of failures\cite{ex1}. Under the assumptions that the network is fully connected and the total number of nodes is known, \cite{ex1} proposes \emph{reducing}, \emph{selecting} and \emph{average} operations and then presents two consensus protocols in a synchronous and an asynchronous environment, separately. In~\cite{ex6}, Abraham et al. improve the protocol proposed in~\cite{ex1}: only $3f+1$ nodes are needed in an asynchronous environment.

Azadmanesh et al.  extend approximate consensus to partially connected networks~\cite{ex8,ex9}. However without using flooding, they did not completely achieve global convergence. Approximate consensus problem is also addressed in multi-agent system \cite{ex5,ex10,ex12,ex2,ex3}. These protocols are called linear iterative consensus and mainly based on linear control theory and matrix theory. Without Byzantine failure, \cite{ex5} indicates that in an undirected graph a sufficient and necessary condition for \emph{convergence} consists in having adequate \emph{joint} connected graphs. For a directed graph, \cite{ex10} points out that a sufficient and necessary condition consists in having a spanning tree contained in adequate \emph{joint} connected graphs. When no Byzantine failure occurs, the speed of \emph{convergence} was analyzed in \cite{ex12}. Based on the knowledge of the global topology, \cite{ex2} and \cite{ex3} address approximate consensus problem in systems where nodes suffer from Byzantine faults.

Without flooding and global topology information, to our knowledge, \cite{ex13} is the first paper where a solution to the approximate Byzantine consensus problem based on the linear iteration method is proposed. A sufficient condition on the network topology is proposed. When this condition is satisfied,  \emph{convergence} is ensured.

While~\cite{ex13} only shows a sufficient condition, \cite{ex7} and \cite{ex15} define a sufficient and necessary condition almost simultaneously. Their new arguments are also related to topology. Yet their conditions are static and can not be adapted directly to mobile environments.

Convergence and gathering problems in environments with mobile robots are also similar with approximate consensus. Each robot needs to make the next moving action according to the results returned by its sensors~\cite{ex17}. However they did not consider any topology requirements: each robot can sense all the other ones.

\section{Sufficient \& Necessary Condition}
\label{sec:condition}

The sufficient and necessary conditions proposed in previous works~\cite{ex7,ex13,ex15} consider a static topology. In our mobile system, the topology is not fixed and changes each time  a node moves. From \emph{Corollary 1}, we know that each time a common new starting round $r$ is reached, $v_{max}$ can no more increase and $v_{min}$ can no more decrease in the future. But, for example, if the network is partitioned into two disconnected sub-networks, an approximate agreement cannot be reached: nodes that belong to the first group may converge to  a value $v1$  while the others may converge to a value $v2$. Even if, during each round $r$, $v_{max}(r)$ can continue to decrease or $v_{min}(r)$ can continue to increase, this does not guarantee that the convergence property will be satisfied.

In this paper a sufficient and necessary condition is proposed. This condition is compatible with the fact that the topology is always changing. Moreover, by its very definition, the proposed condition consider the dynamic evolution of the distribution of values within the system. More precisely, it focuses on the particular correct nodes that have currently either  the value $v_{min}$ or the value $v_{max}$. At least one of these nodes has to receive from its neighborhood enough messages (quantity constraint)  that contain values different from its own current value (quality constraint).

To formally define what is expected in terms of quality, we first provide the definition of a {\it proper value}.  The problem definition (See Section~\ref{sec:model}) refers to a parameter $\epsilon$  which sets the level of precision that needs to be obtained to consider that an agreement is reached. We introduce a second parameter called $\delta$ whose range of possible values is $(0, \frac{\epsilon}{2}]$. This parameter, which is not used in the protocol, is necessary to define the condition and the notion of proper value on which the condition relies.
This non-zero positive integer (whose value may be very small) allows us to define five intervals of values as follows. When a common new starting round $r \in S^{c}$ begins, all correct nodes have values in the range $[v_{min}(r), v_{max}(r)]$. Five value intervals are defined.
 \emph{Minimum value} corresponds to the value $v_{min}(r)$. \emph{Maximum value} corresponds to the value $v_{max}(r)$. \emph{Nearly minimum value} represents the value interval $(v_{min}(r),v_{min}(r)+\delta)$.  Symmetrically, \emph{Nearly maximum value} represents the value interval $(v_{max}(r)-\delta, v_{max}(r))$. Finally, \emph{Middle value} represents the value interval $[v_{min}(r)+\delta,v_{max}(r)-\delta]$. These three intervals are defined at the beginning of a new phase ({\it i.e.} just before a round $r=kR_c+1$ begins) and will not change during $R_c$ consecutive rounds. During a round $r'$ of the phase $k$ ($r \leq  r' < r+R_{c}$), a correct node $p_{i}$ may change its value $v_{i}$. Depending on its value $v_{i}(r')$, a correct node is classified in one of the following five groups:  $CMin(r')$,  $CNin(r')$, $CMid(r')$, $CNax(r')$,  and $CMax(r')$. The letter $C$ at the beginning of the name of a group indicates that the members of the group are correct nodes. Throughout a phase $k$, the same interval of value is associated to a group. During a round $r'$, the rule for assigning a correct node $p_{i}$ to a group is simple: the value $v_{i}(r')$ must belong to the corresponding interval. During a round, a correct node belongs to exactly one group. By definition, the distribution of the nodes into the five  groups may change at each round of a phase.  Figure 4(a) summarize the above discussion.

\begin{figure}[ht]
    \centering
    \includegraphics[width=2.8in]{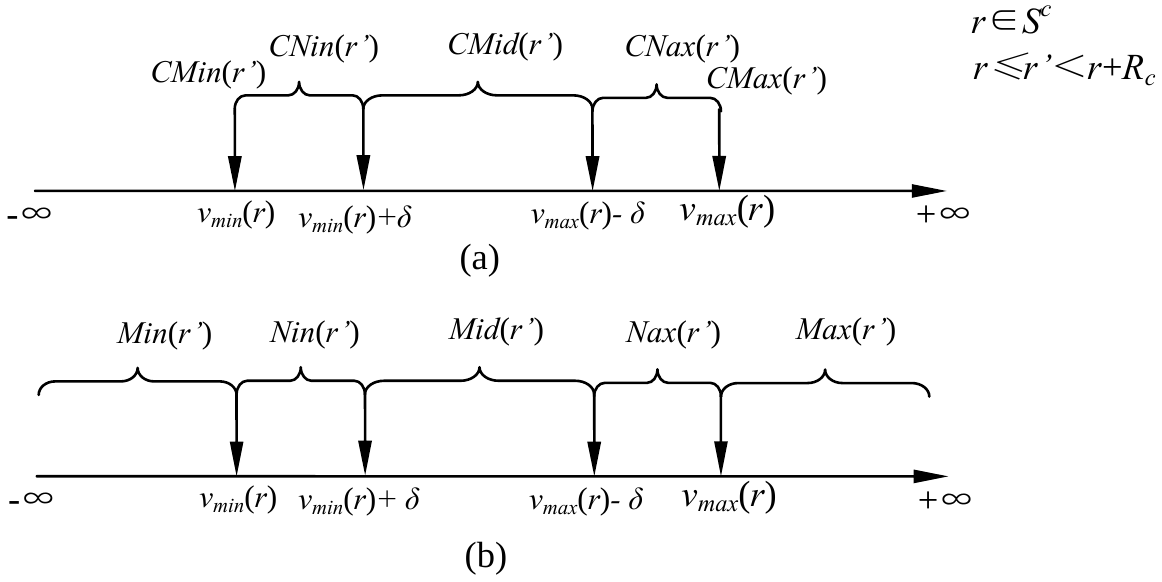}
    \caption{Value intervals and node sets}
    \label{Fig4}
\end{figure}

Byzantine nodes may exist in the system. They can propose values that belong to $[v_{min}(r), v_{max}(r)]$ but also values that are beyond this interval. If Byzantine nodes propose legal values, they also make sense. We use $Max(r')$, $Min(r')$, $Nax(r')$, $Nin(r')$ and $Mid(r')$ to represent groups that mix correct nodes and Byzantine nodes. Due to the Byzantine nodes, the value intervals corresponding to $Min(r')$ and $Max(r')$ are enlarged respectively to $(-\infty, v_{min}(r)]$ and $[v_{max}(r), +\infty)$. Of course, $CMax(r') \subseteq Max(r')$,  $CNax(r')\subseteq Nax(r')$, $CMid(r')\subseteq Mid(r')$, $CNin(r')\subseteq Nin(r')$, and $CMin(r') \subseteq Min(r')$. Node that a Byzantine node can belong to several group during the same round. These five new groups are depicted in Figure 4(b). Except $Max$ and $Min$, the three other groups can be empty. $Nin$ and $Nax$ are two special groups which are used to distinguish legal values that are sufficiently different from either the current minimal value $v_{min}$ or the current maximal value $v_{max}$.

The proposed condition only affects the nodes that have either the minimum value or the maximum value when a phase $k$ begins. Let $r = kR_c +1$ be the first round of this phase. The  targeted nodes belong to  $CMin(r)$ or $CMax(r)$.

\begin{definition}
For any common new starting point $r \in S^{c}$ and for any round $r'$ such that  $r \leq r'< r+R_{c}$,  a proper value, from the point of view of a correct node that belongs to  $CMin(r)$ is a value $v_{j}(r')$ that belongs to the interval $[v_{min}(r)+\delta, +\infty)$ while a proper value, from the point of view of a correct node that belongs to  $CMax(r)$ is a value $v_{j}(r')$ that belongs to the interval $(-\infty, v_{max}(r)-\delta]$.
\end{definition}

Note that a proper value is not necessarily a legal value. A proper value can be a fake value proposed by a Byzantine node.
Based on the above definition, we can now express the proposed condition.

\begin{theorem}
The \emph{convergence} property is satisfied by the proposed protocol if the following sufficient condition always holds:\\
$\forall r \in S^{c}$ such that  $v_{max}(r) - v_{min}(r) \geq \epsilon$,\\
$\exists p_{i} \in C_n$ such that $p_{i} \in CMax(r) \cup CMin(r)$,\\
$\exists r'$ such that  $r \leq r'< r+R_{c}$,\\
$\exists V_{q} \in V$ such that $V_{q} \subseteq JN_{i}^{r'}$,\\
$\mid V_{q} \mid \geq f+1$ and $\forall p_{j}\in V_{q}$, $v_{j}(r')$ is a proper value.
\end{theorem}

As long as the convergence test is not satisfied, for each phase (characterized by its associated common new starting round $r$), at least one correct node $p_{i}$ among those which are members of the groups $CMin(r)$ or $CMax(r)$ must, at least once during the phase ({\it i.e.} during a round $r'$), compute a new value using $f+1$ (quantity constraint) proper values (quality constraint) received during the current phase. As mentioned before, a proper value is not a legal value. But, due to the fact that at least $f+1$ are gathered by $p_{i}$, at least one of them is not removed during the reducing procedure and is a legal value.
Note that the above condition does not ensure that either $v_{min}$ increases or $v_{max}$ decreases during a phase. In fact, several correct nodes may have the minimum value or the maximum value when the phase begins. The condition just ensures that at least one of them will increase or decrease its value.

Now the reason why $Nax$ and $Nin$ have been defined is explained. The requirements expressed in \emph{Theorem 2} only focus on the nodes of  the groups $CMin(r)$ and $CMax(r)$: the other correct nodes are not concerned.  Suppose $v_{max}(r) - v_{min}(r) \geq \epsilon$ and values from $Nax$ and $Nin$ are not excluded. In this situation, if the nodes that belong to $CMin(r)$ and $CMax(r)$ only receive values from respectively $(v_{min}(r), v_{min}(r)+\delta)$ and $(v_{max}(r)-\delta,v_{max}(r))$ and if all the other correct nodes do not change their values, then there may exist a value $\mu \geq \epsilon$, such that $\lim\limits_{r\rightarrow +\infty} v_{max}(r)-v_{min}(r)=\mu$ .

Mobility has a strong impact on the fact that the condition can be satisfied or not. If some correct nodes are always moving far away, convergence can not be obtained. In fact, a correct node can remain isolated as long as it is neither in $CMax(r)$ nor in $CMin(r)$. But, of course, in many applications, a correct node can not always determine if it is currently concerned or not by the condition. The fact that the nodes in $CMin(r)$ or $CMax(r)$ can obtain enough proper values depends not only on the trajectory and the speed of the correct nodes. It depends also on the cardinality of the system ({\it i.e.} the cardinality of the five sets). For example, if the cardinality of the union set $\bigcup\limits_{r'}(Max(r')\cup Nax(r')\cup Mid(r'))$ ($r'\in\{r,r+1,...,r+R_{c}-1\}$) is smaller than $f+1$, no node of $CMin(r)$ has the possibility to meet both the quantity and quality constraint. However, in that case, the cardinality of the union set $\bigcup\limits_{r'}(Min(r')\cup Nin(r')\cup Mid(r'))$ should be sufficient to ensure that at least one node that belongs to $CMax(r)$ can collect enough proper values.

\begin{lemma}
To ensure that at least one node (either in $CMin$ or in $CMax$) can collect enough proper values, the cardinality of the system must satisfy the following constraint: $n\geq 3f+1$.
\end{lemma}

\begin{proof}
Suppose that there is only $3f$ nodes in the network. Think about this situation $\mid CMax\mid=1$, $\mid CMin\mid=1$, $\mid CNax\mid=f-1$,$\mid CNin\mid=f-1$ and $\mid CMid\mid=0$. Moreover,  the $f$ Byzantine nodes propose values bigger than $v_{max}$ to the nodes in $CMax$ and, at the same time, the $f$ byzantine nodes propose values smaller than $v_{min}$ to nodes in $CMin$. In that case only $f$ proper values can be seen by nodes in $CMax$ and $CMin$. The cardinality of the system is not sufficient to provide "quantity" and "quality" simultaneously.

While $v_{max} - v_{min} \geq \epsilon$, if $n=3f+1$, there is always a chance to satisfy at least one node in $CMax$ or $CMin$. The proof is by contradiction. Suppose there is no chance to gather enough proper values neither for the nodes in $CMax$ nor for those in $CMin$. Suppose Byzantine nodes send $f$ illegal values bigger than $v_{max}$ to nodes in $CMax$ and send $f$ values smaller than $v_{min}$ to nodes in $CMin$ or just keep silent. In that way the Byzantine nodes do not contribute to the satisfaction of the condition. The remaining $2f+1$ are all correct nodes. Suppose $\mid CMid \mid=0$, because any nodes belongs to $CMid$ helps both the nodes of $CMin$ and $CMax$ to satisfy the constraints.  So according to the pigeonhole principle, at least one of the following two inequalities must be true:
$\mid CMax + CNax\mid\geq f+1$ or $\mid  CMin + CNin\mid\geq f+1$.
A contradiction.
\end{proof}

Let us now consider that the system is populated with a sufficient number of correct nodes:  $n\geq 3f+1$. Even if the nodes travel arbitrarily within the system, we assume that the condition is always satisfied. First we prove two general lemmas related to the convergence property. By definition, the convergence property is a stable property. Once the convergence is reached, this property  remains true.

\begin{lemma}
\label{lemmax}
Let $r$ be a common new starting point. If $v_{max}(r) - v_{min}(r) < \epsilon$ then convergence is already reached when round $r$ starts.
\end{lemma}

\begin{proof}
Let us assume that $v_{min}(r)$ is proposed by a correct node $p_{i}$ while  $v_{max}(r)$ is proposed by a correct node $p_{j}$. By definition, for any correct node $p_{k}$, $v_{min}(r) \leq v_k(r) \leq v_{max}(r)$. As $r \in S^{C}$, due to Corollary 1, $\forall r' \geq r$, $v_{min}(r) \leq v_k(r') \leq v_{max}(r)$. Therefore, $\forall r' \geq r$, $v_{min}(r) \leq v_{min}(r')$ and $v_{max}(r') \leq v_{max}(r)$. Consequently, as $v_{max}(r) - v_{min}(r) < \epsilon$,  $\forall r' \geq r$, $v_{max}(r') - v_{min}(r') < \epsilon$.  Thus convergence is already reached when round $r$ begins.
\end{proof}

Note that the fact that $r$ is a common starting round is essential in the proof of Lemma~\ref{lemmax}. If $r$ is not a common starting round, it could be the case that $v_{max}(r) - v_{min}(r) < \epsilon$ while $v_{max}(r+1) - v_{min}(r+1) \geq \epsilon$.

\begin{lemma}
\label{lemmay}
When a common new starting round $r$ begins, convergence is already reached if and only if either $CMin(r) = CMax(r)$ or $CNax(r) \cap CNin(r) \neq \emptyset$.
\end{lemma}

\begin{proof}
By definition, if convergence is already reached during a round $r$ (that belongs or not to $S^{c}$), either all the values of the correct nodes are equal or they differ by at most $\epsilon$. In the first case, $CMin(r) = CMax(r)$ and $CNax(r) = CNin(r) = \emptyset$. In the second case, $CMin(r) \neq CMax(r)$ and $v_{max}(r) - v_{min}(r) < \epsilon$. As $ \delta \leq \frac{\epsilon}{2}$, we have $v_{max}(r)-\delta < v_{min}(r)+\delta$. Thus $CNax(r) \cap CNin(r) \neq \emptyset$. The first implication holds.
To prove the second implication, let us first consider that $CMin(r) = CMax(r)$. Due to Corollary 1, all the correct nodes will keep the common value in the future. Now if $CNax(r) \cap CNin(r) \neq \emptyset$, there exist at least one correct node $p_{i}$ who has proposed a value $v_i(r)$ which belongs both to $(v_{min}(r),v_{min}(r)+\delta)$ and $(v_{max}(r)-\delta, v_{max}(r))$.  Thus $v_{max}(r)-\delta < v_{min}(r)+\delta$.  By assumption, $\delta$ belongs to $(0,\frac{\epsilon}{2}]$. Therefore, $v_{max}(r) - v_{min}(r) < 2 \delta \leq \epsilon$. Again, due to Corollary 1, convergence is already reached when round $r$ begins.
\end{proof}

\begin{lemma}
\label{lemmaz}
Let $r \in S^{c}$ be a common new starting round such that convergence is not yet reached when $r$ starts. Let $\omega_{1}$  be a positive integer ($\omega_{1} \geq 1$). Let $\varsigma_1$ and  $\varsigma_2$ be two reals that belong to $(0, 1]$ and such that:\\
$v_{min}(r+ \omega_{1} R_{c}) = \varsigma_{1} v_{min}(r)$\\
$v_{max}(r+ \omega_{1} R_{c}) = \varsigma_{2} v_{max}(r)$\\
During the $\omega_{1}$ phases, if neither the minimal value increases ($\varsigma_1 = 1$) nor the maximal value decreases ( $\varsigma_2 = 1$) then the two following predicates are both satisfied:
\begin{enumerate}
\item $(\mid CMin(r + \omega_{1} R_{c}) \mid < \mid CMin(r) \mid)$ or  $(\mid CMax(r + \omega_{1} R_{c}) \mid < \mid CMax(r) \mid)$
\item $\omega_{1} < n$ 
\end{enumerate}
\end{lemma}

\begin{proof}
Due to Corollary 1, it is obvious that we can rewrite $v_{min}(r+R_{c})$ and $v_{max}(r+R_{c})$ using the defined $\varsigma_1$ and $\varsigma_2$. Let us assume that the minimal and the maximal values are  stable during the $\omega_{1} R_{c}$ rounds: $\varsigma_1 = 1$ and  $\varsigma_2 = 1$. We demonstrate that, after each phase, the cardinality of at least one of the two sets decrease. The fact that the property holds when $\omega_{1} = 1$ allows us to conclude that the property holds for any value of  $\omega_{1}$. During the $R_{c}$ numbered $r$, $r+1$, $\ldots$, $r+R_{c}-1$, due to the necessary condition expressed in Theorem 2, there exists at least one round during which a "good" phenomena occurs. Let us consider the highest round $r'$ during which the condition is true and let $p_{i}$ be a node such that $p_{i}$ has gathered enough proper values:
$\exists V_{q}\subseteq JN_{i}^{r'}$ such that $\mid V_{q}\mid\geq f+1$ and $\forall p_{j}\in V_{q}$, $v_{j}(r')$ is a proper value. Without loss of generality, let us assume that $p_{i}$ belongs to $CMin(r)$. Due to Corollary 1, $v_{i}(r') \geq v_{min}(r)$. Moreover, due to the condition, the reducing procedure and the average procedure are executed by node $p_{i}$. During the reducing procedure, at least one value greater or equal to $v_{min}(r)+\delta$ is not removed. In the worst case, all the other values used during the computation are equal to $v_{min}(r)$. Even in that case, the new computed value of $p_{i}$ is such that $v_{i}(r'+1) > v_{min}(r)$. If $r'$ is not the last round of the phase (and despite the fact that r' is the highest round of the phase during which the condition is true), it could be the case that $p_{i}$ computes again its new value during rounds $r''$ such that $r' < r'' < r+R_{c}$. In the worst case, $p_{i}$ will compute the average between its own value ($> v_{min}(r)$) and a set of gathered values all equal to $v_{min}(r)$. Thus, when the next phase begins, $v_{i}(r+R_{c}) > v_{min}(r)$. Consequently, if $v_{min}(r+R_{c}) = v_{min}(r)$, at least one node (namely $p_{i}$) belongs to $CMin(r)$ but not to $CMin(r+R_{c})$. A similar reasoning can be adopted if $p_{i}$ belongs to $CMax(r)$. As the number of nodes is finite, and as at least one node per phase is removed from either $CMin$ or $CMax$, we can conclude that $v_{min}$ and $v_{max}$ both remain stable during at most $n-1$ phases if convergence was not yet reached during round $r$. It is always the case that $v_{min}(r+ n R_{c}) > v_{min}(r)$ or $v_{max}(r+ n R_{c}) < v_{max}(r)$. Thus the second predicate $\omega_{1} < n$ also holds.

\end{proof}

We prove now Theorem 2.

\begin{proof}
(\emph{Theorem 2}) 

Let $r$ be a starting round during which convergence is not yet achieved.
From Lemma~\ref{lemmaz}, we can conclude that there exist a positive integer $\omega_{1}$ such that either $v_{min}(r+ \omega_{1} R_{c}) = \varsigma_{1} v_{min}(r)$ with $\varsigma_1 \in (0, 1)$ or $v_{max}(r+ \omega_{1} R_{c}) = \varsigma_{2} v_{max}(r)$ with $\varsigma_1 \in (0, 1)$. Let $d = v_{max}(r)-v_{min}(r)$. In the first case we have:\\
$d > v_{max}(r+\omega_{1} R_{c})-v_{min}(r)$.\\
In the second case, \\
$d >v_{max}(r)-v_{min}(r+\omega_{1} R_{c})$.
Due to Corollary 1, $v_{min}(r) \leq v_{min}(r+\omega_{1} R_{c})$ and $v_{max}(r) \geq v_{max}(r+ \omega_{1} R_{c})$. Therefore, in both cases:
$d > v_{max}(r+\omega_{1} R_{c}) - v_{min}(r+\omega_{1} R_{c})$.

The difference $v_{max} - v_{min}$ will always decrease. Yet this is not sufficient to prove that eventually convergence is reached. To prove this last point, we have to show the existence of a finite integer $\tau$, such that: $v_{max}(\tau R_{c}+1)-v_{min}(\tau R_{c}+1)<\epsilon$.

The proof is by contradiction. Suppose that the above condition is never satisfied. In other words, whatever the value of the integer $\tau$, the difference $v_{max}(\tau R_{c}+1)-v_{min}(\tau R_{c}+1)$ only approaches a value $\mu$ but $\mu \geq \epsilon$. Figure~\ref{Fig5} depicts such a scenario. In this representation we assume that there exists always a real $\Delta$ such that:
$\Delta \geq 0$ and $v_{max}(\tau R_{c}+1)-v_{min}(\tau R_{c}+1) = \Delta + \mu$.

In Figure~\ref{Fig5}, $\Delta$ is represented by the sum of  $\Delta_{1}$ and $\Delta_{2}$.

\begin{figure}[ht]
    \centering
    \includegraphics[width=2.5in]{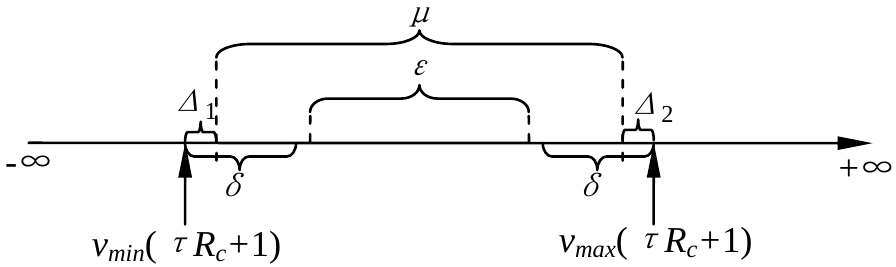}
    \caption{Example of value $\mu$}
    \label{Fig5}
\end{figure}

Due to Lemma~\ref{lemmaz}, when $\tau$ approaches to infinity, $\Delta$ approaches to 0. We will show that $\mu$ is not a limit: the  value of $\Delta$ may become less than zero and never become positive again. First we will show that there exists a particular positive value of $\Delta$ such that after a single execution of the average procedure during a round $r$, the value of $\Delta$ becomes negative. Then we show that there exists a particular positive value of $\Delta$ such that after a finite number of execution of the procedure average, the value of $\Delta$  remains negative forever.

First let us consider a particular phase $k$. Let  $r = (k-1) R_{c} + 1$ be the first round of this phase. Let $d$ denotes the difference $v_{max}(r) - v_{min}(r)$. Let us assume that at the beginning of phase $k$, the value of $\Delta$ is positive and equal to $d - \mu$. Due to Lemma~\ref{lemmaz}, there exists a round $r'$ such that: 
 $CMin(r) = CMin(r')$ and $CMax(r) =CMax(r')$ and $(CMin(r) \neq CMin(r'+1)$ or $CMax(r) \neq CMax(r'+1))$. 
 
 Let $d'$ be the difference $v_{max}(r'+1)-v_{min}(r'+1)$. We have $d > d'$. The value of $\Delta$ decreases by $d-d'$ between round $r$ and round $r'+1$. The value of $\Delta$ is equal to $d' - \mu$ during round $r'+1$. This value is negative if $d' < \mu$.  First we compute an estimation of the minimal value $d-d'$  that can be observed. Obviously, to ensure that the difference $d-d'$ is as small as possible, either just the minimal value has to increase or just the maximal value has to decrease (but not both during the same round $r'$). As the two cases are symmetric, let us consider that the minimal value increases while the maximal one remains stable. Let us consider a node $p_{i}$ such that $p_{i}$ has the minimal value during round $r'$. This node may have again the minimal value during round $r'+1$. In that case,  $v_{min}(r'+1) > v_{min}(r')$ and $d-d' = v_{min}(r'+1) - v_{min}(r')$. Our aim is to obtained an estimation (more precisely an under-estimation) of the difference $d-d'$.  To be allowed to compute a new value, the node $p_{i}$ must gather at least $f+1$ proper values (in the worst case, these value can be equal to $v_{min}(r') + \delta$). After the reducing procedure, $p_{i}$ keeps at most $n-f$ values. Furthermore at least one correct node propose a value greater than $\mu$ (otherwise this contradict the fact that the limit $\mu$ is respected between round $r$ and $r'$). Yet as our goal is just to provide an under-estimation, we consider an (unrealistic) worst case. The new value of node $p_{i}$ computed during round $r'$ during the execution of the average procedure is the average between $n$ values where $n-1$ are equal to $v_{min}(r')$ and a single one is a proper value (more precisely, the minimal proper value, namely $v_{min}(r') + \delta$). In that case, we have:\\
$v_{min}(r'+1) > \frac{(n-1)v_{min}(r')+(v_{min}(r')+\delta)}{n}$\\
As the right part of the above formula is an under-estimation, we use the symbol "$>$" rather than the symbol "$\geq$". Based on the previous formula, we conclude that:\\
$d-d' > v_{min}(r') + \frac{\delta}{n} - v_{min}(r')$

We have proved the existence of a particular positive value of $\Delta$ (during round $r'$) that is small enough to imply that the value of $\Delta$ can be negative during the next round.
If during round $r'$, the value of $\Delta$ (which is equal to $d - \mu$) is strictly less than $\frac{\delta}{n}$, then the value of $\Delta$ can be negative during round $r'+1$. This contradict the fact that $\mu$ is a limit for $v_{max}-v_{min}$ which is never 
violated.

At this stage, we have just demonstrated that $\mu$ is not a limit. We now show that after some time, this limit will be breached permanently. To understand why a violation of the limit $\mu$ is sometimes transient, let us consider that during a phase a single node $p_{i}$ has the smallest value $v_{min}(r)$. During the first round of this phase, it may broadcast this value to all the other nodes and then compute a new value $v_{min}(r')$ which violates the limit $\mu$. Unfortunately any other node may compute again a new value based on its own value and old values contained in their log that are less that $v_{min}(r')$ and possibly very closed from $v_{min}(r)$. As a consequence the value of $p_{i}$ may decrease again and respect again the limit $\mu$.

Let us consider a common new starting round $r$ such that the limit $\mu$ is respected. Due to Lemma~\ref{lemmaz}, after $n-1$ phases, either all the nodes that have the minimal value or all 
the nodes that have the maximal value during round $r$ have now adopted either an higher value or respectively a smaller value. Again without loss of generality, let us consider the worst case where all the nodes (except one)  where sharing the minimal value $v_{min}(r)$ during round $r$ while a single node has a value equal to $v_{min}(r)+\mu$. Again, our goal is to identify a limit $x$ (even if this one is under-estimated) that shows that at a beginning of round $r+ n R_{c}+1$, no value less than $v_{min}(r) + x$ remains in the system. Again, during each phase, at least one node $p_{i}$ ,which has the minimal value when the phase begins modifies its value and adopts during a round of the phase, a value which is at least equal to $v_{min}(r)+ \frac{\delta}{n}$ (See the above discussion). In the worst case, this change occurs during the first round of the first phase denoted $r$. Then $p_{i}$ may compute again its value during the next $n R_{c} - 1$ following rounds. If it receives only values that are equal to $v_{min}$ during these rounds, its value at the end of the phase is strictly greater than $v_{min}(r')+\frac{\delta}{n^{n R_{c}}}$. Once this last phase ends, the value of $p_{i}$ can no more decrease. Therefore, after at most $n$ phases, all the correct nodes have a value that will remain greater than $v_{min}(r)+\frac{\delta}{n^{n R_{c}}}$. Consequently there exits a positive value of $\Delta$ such that the violation of the limit $\mu$ is permanent. 

The condition is sufficient to ensure convergence.
\end{proof}

Regarding the fact that the condition is necessary, we identify a weaker condition. Indeed, the condition does not have to be satisfied in each phase but only infinitely often. This modification of the condition has no major impact on the way we prove that the condition is a sufficient condition. Some properties are no more satisfied "at the end of each phase" but "after a finite number of phases".

To prove that the condition is necessary, we show that the quantity constraint and the quality constraint are both necessary. Within the set of $n$ nodes, let us assume that half of the $n-f > 2f$ correct nodes share a same value $v_{min}$ while the second half share the value $v_{max}$. We assume that $v_{min}$ and $v_{max}$ are such that the convergence is not yet reached. If a correct node gathers only $f$ proper values before computing its new value, it could be the case that the values that remain after the execution of the reducing procedure are all equal to its own value. Thus no correct node will change its value. If a correct node gathers $f+1$ values but at least one of them is not a proper value, it is also possible that all the proper values will be removed during the reducing procedure. Again, the values of the correct nodes will be stable and convergence is not ensured.

\section{Conclusion}
\label{sec:conclusion}

In this paper, we addressed approximate Byzantine consensus problem in partially connected mobile networks. An architecture for both moving and consensus protocol has been proposed. Then an approximate consensus protocol based on a linear iteration method has been presented. In order to take advantage of mobility, in this protocol, nodes are allowed to collect messages during at most $R_{c}$ consecutive rounds. Afterwards, we have defined a sufficient and necessary condition that allows to satisfy \emph{convergence}. Compared to existing papers, this novel condition is dynamic and not "universal". It only focuses on the correct nodes which propose the maximum or the minimum value and requires that, from time to time, at least one of them should receive enough (quantity constraint) proper (quality constraint) values. Our analysis shows that if $n\geq3f+1$, the condition has chances to be satisfied and consensus can be reached. We are now working on particular mobility scenarios where either the existence of some meeting points or a predefined trajectory and scheduling  allow to prove that the condition is satisfied. Simulations are also conducted to analyze the impact of a tuning of the $R_c$ parameter.

{\scriptsize
\section*{Acknowledgment}
This work is partially supported by Natural Science Foundation, China under grant 60973122 and National 863 Hi-Tech Program, China under grant 2011AA040502. This work is partially supported by the ANR French national program for Security and
Informatics (grant \#ANR-11-INSE-010, project AMORES).
}

\bibliographystyle{unsrt}
{\scriptsize
\bibliography{IEEEabrv,disctemplate}
}


\tableofcontents
\end{document}